\documentclass[prl,aps,showpacs,
twocolumn,
superscriptaddress]{revtex4-1}
\usepackage{mathrsfs}
\usepackage{amsthm}
\usepackage{amsmath}
\usepackage{amssymb}
\usepackage{graphicx}

\makeatletter
\def\Hy@safe@activestrue{}
\makeatother

\def\mod{\mathop{\rm mod}}

\def\mat#1{\mathcal{#1}}

\newtheorem{theorem}{Theorem}
\newtheorem{consequence}[theorem]{Corollary}
\newtheorem{lemma}{Lemma}

\begin{document}

\title{Fault-Tolerance of "Bad" Quantum Low-Density Parity Check Codes }

\author{Alexey A. Kovalev}

\affiliation{Department of Physics \& Astronomy, University of California,
  Riverside, California 92521, USA}

\author{Leonid P. Pryadko}

\affiliation{Department of Physics \& Astronomy, University of California,
  Riverside, California 92521, USA}

\begin{abstract}
    Quantum low-density parity check (LDPC) codes such as generalized toric codes with finite rate suggested by Tillich and Z\'emor offer an
alternative route for quantum computation. Here, we study LDPC codes and
  show that any family of LDPC codes, quantum or classical, where
  distance scales as a positive power of the block length, has a finite error threshold. Based on that, we conclude that quantum LDPC codes, for sufficiently
  large quantum computers, can offer an advantage over the toric codes.
\end{abstract}

\date{\today}

\pacs{03.67.Lx, 03.67.Pp,  64.60.ah}

\maketitle

A practical implementation of a quantum computer will rely on quantum error correction (QEC) \cite{shor-error-correct,Knill-Laflamme-1997,Bennett-1996} due to the fragility of quantum states.
There is a strong belief that surface (toric) codes \cite{kitaev-anyons,Dennis-Kitaev-Landahl-Preskill-2002} can offer the fastest route to scalable quantum computer due to the error threshold exceeding 1\% and locality of required gates \cite{Raussendorf-Harrington-2007,Wang-Fowler-Austin-Hollenberg-Lloyd-2011,Bombin-PRX-2012}. Unfortunately, in the nearest future,
the surface codes (in fact, any two-dimensional codes with local stabilizer
generators\cite{Bravyi-Poulin-Terhal-2010}) can only lead to proof of the principle realizations as they encode a limited number of qubits (k), making a practical implementation of a useable
quantum computer expensive in terms of the required number $n$ of physical qubits (e.g. 220 million physical qubits are necessary for a useful realization
of Shor's algorithm \cite{Fowler:2012arXiv}).

A large family of quantum low-density
parity-check (LDPC) codes (a non-local generalization of toric codes) has been
constructed by Tillich and Z\'emor\cite{Tillich2009}.  These quantum
hypergraph-product codes (QHPCs) contain families of CSS codes with finite asymptotic rate $R\equiv k/n$
which substantially improves upon toric codes where $R=0$ (see Fig.~\ref{fig:Visualization}).
This construction can also be modified to generalize the
rotated toric codes\cite{Kovalev-PRA2011} (e.g., checkerboard codes) with a finite-factor rate
improvement\cite{Kovalev-arxiv2012}.  Just as for the toric
codes, the distance of QHPCs scales as a square root of the block length,
$d\propto n^{1/2}$. 

In general, removing the restriction of locality should considerably improve the code parameters. Non-local two-qubit gates are relatively
inexpensive with floating gates \cite{Loss:PRX2012}, superconducting and trapped-ion qubits, as well as more
exotic schemes with teleportation\cite{yamamoto-cnot-2003,%
  Martinis-science-2005,Benhelm-2008,Friedenauer-2008,%
  Bennett-teleportation-1993,%
  Gottesman-Chuang-1999}. Thus, one can consider a much wider class of quantum LDPC codes\cite{Postol-2001,MacKay-Mitchison-McFadden-2004} for which,
compared to general quantum codes, each quantum measurement involves fewer qubits,
measurements can be done in parallel, and also the classical processing can be
enormously simplified. 

\begin{figure}[htbp]
\centering \includegraphics[width=0.48\columnwidth]{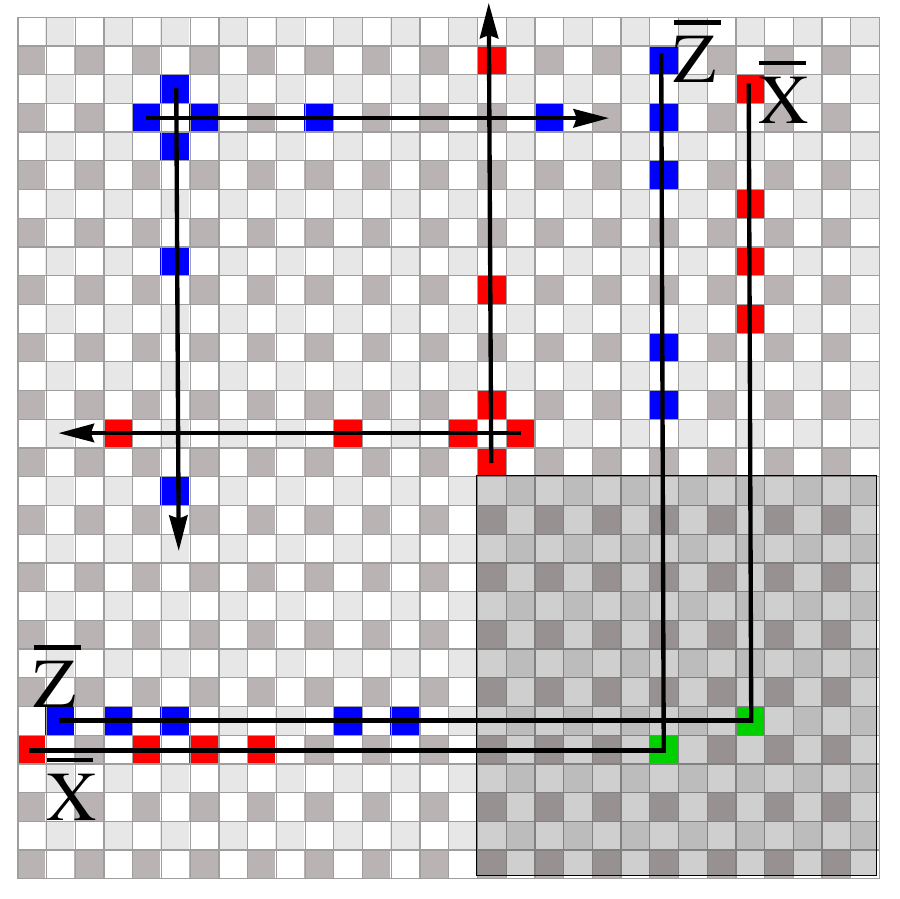}\hskip0.1in
\includegraphics[width=0.48\columnwidth]{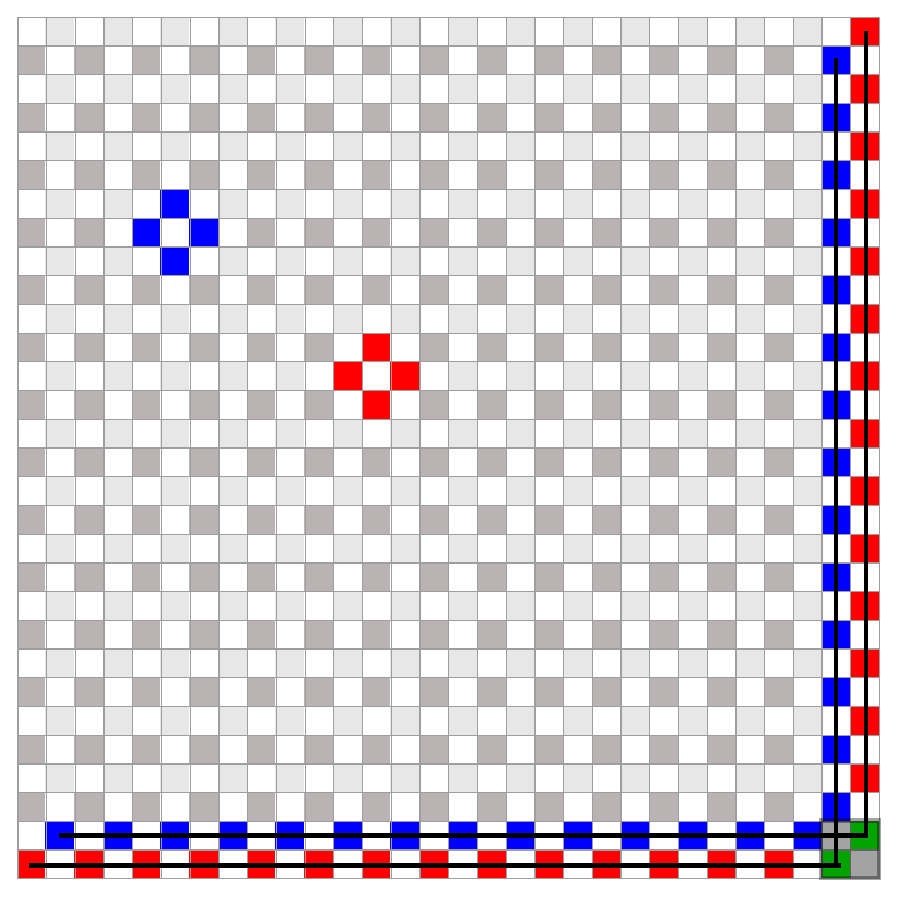} \caption{(Color online) Left: Two stabilizer generators (marked by arrows)
and two pairs of anticommuting logical operators (marked by lines)
of a $[[450,98,5]]$ code in Eq.~(\ref{eq:Till}) formed by circulant
matrices $\mathcal{H}_{1}=\mathcal{H}_{2}$ with the first row $[1,1,0,1,0,0,0,1,0\dots 0]$ (red -- $X$ operators, blue
-- $Z$ operators, green -- overlap of $Z$ and $X$ operators, dark
and light gray -- dual sublattices of physical qubits). Other
stabilizer generators are obtained by shifts over the same sublattice
with periodic boundaries. Shaded regions: each gray square uniquely
corresponds to a pair of logical operators, thus $98$ encoded logical
qubits. Right: same for the toric code $[[450,2,15]]$.}
\label{fig:Visualization} 
\end{figure}

Furthermore, at sufficiently large blocklength, a fault tolerant family of quantum LDPC codes with
a finite asymptotic rate will require fewer physical
qubits compared to a realization based on copies of toric code. Achieving small computational error requires use of codes with large
blocklength, thus, realization of a quantum computer based on a finite-rate quantum LDPC code can reduce the required number of physical qubits. 
Unfortunately, the parameters as well as fault tolerance of general quantum LDPC codes are largely unexplored.

In this Letter, we discuss
error-correction properties and fault-tolerance of families of finite-rate quantum
(and where noted classical) LDPC codes whose relative distance $\delta\equiv d/n$ tends to
zero in the limit of large blocklength. Even though we term such codes as "bad" (as good codes should have finite rate and finite relative distance $d/n$, see Ref.~\cite{Calderbank-Shor-1996}) these are the best finite-rate quantum LDPC codes with explicitly known distance.  For random uncorrelated (qu)bit errors (\emph{e.g., quantum depolarizing channel}), we establish the
existence and give a lower bound for the single (qu)bit error rate below which
the decoding with probability one is possible, and analyze the scaling of
successful decoding probability with the blocklength. This result is obtained by separating errors into small independent clusters on a graph, a construction analogous to the cluster theorem \cite{Domb-Green-bookX}.  We also give related
bounds for fault-tolerant operation in the presence of syndrome measurement
errors.  A similar analysis for errors when the erroneous (qu)bits are known (\emph{erasure channel}) allows us to
establish an upper limit for the achievable rate of a quantum LDPC code with
power-law scaling of the distance with blocklength.  These results are important since, unlike for regular QEC codes, there are very few
general lower (existence) or upper bounds for quantum LDPC
codes\cite{Delfosse-Zemor-2012}.

\textsc{Definitions.}  A binary linear code $\mathcal{C}$ with
parameters $[n,k,d]$ is a $k$-dimensional subspace of the vector space
$\mathbb{F}_2^n$ of all binary strings of length $n$.  Code distance
$d$ is the minimal weight (number of non-zero elements) of a non-zero
string in the code.  A linear code is uniquely specified by the binary
parity check matrix $H$, namely $\mathcal{C}=\{\mathbf{c}\in
\mathbb{F}_2^n| H \mathbf{c}=0\}$, where operations are done $\mod 2$.

A quantum $[[n,k,d]]$ (qubit) stabilizer code $\mathcal{Q}$ is a
$2^k$-dimensional subspace of the $n$-qubit Hilbert space
$\mathbb{H}_{2}^{\otimes n}$, a common $+1$ eigenspace of all operators in an
Abelian \emph{stabilizer group} $\mathscr{S}\subset\mathscr{P}_{n}$,
$-\openone\not\in\mathscr{S}$, where the $n$-qubit Pauli group
$\mathscr{P}_{n}$ is generated by tensor products of the $X$ and $Z$
single-qubit Pauli operators.  The stabilizer is typically specified in terms
of its generators, $\mathscr{S}=\left\langle S_{1},\ldots,S_{n-k}\right\rangle
$; measuring the generators $S_i$ produces the \emph{syndrome} vector.  The
weight of a Pauli operator is the number of qubits it affects.  The distance $d$ of a quantum code is the
minimum weight of an operator $U$ which commutes with all operators from the
stabilizer $\mathscr{S}$, but is not a part of the stabilizer,
$U\not\in\mathscr{S}$. A code of distance $d$ can detect any error of weight up to $d-1$, and correct up to $\lfloor d/2\rfloor$.

A Pauli operator $U\equiv i^{m}X^{\mathbf{v}}Z^{\mathbf{u}}$, where
$\mathbf{v},\mathbf{u}\in\{0,1\}^{\otimes n}$ and
$X^{\mathbf{v}}=X_{1}^{v_{1}}X_{2}^{v_{2}}\ldots X_{n}^{v_{n}}$,
$Z^{\mathbf{u}}=Z_{1}^{u_{1}}Z_{2}^{u_{2}}\ldots Z_{n}^{u_{n}}$, can be
mapped, up to a phase, to a quaternary vector, $\mathbf{e}\equiv
\mathbf{u}+\omega \mathbf{v}$, where $\omega^2\equiv
\overline{\omega}\equiv\omega+1$.  A product of two quantum operators
corresponds to a sum ($\mod2$) of the corresponding vectors.  Two Pauli
operators commute if and only if the \emph{trace inner product} $\mathbf{e}_1
* \mathbf{e}_2 \equiv \mathbf{e}_1 \cdot \overline{\mathbf{e}}_2 +
\overline{\mathbf{e}}_1 \cdot \mathbf{e}_2$ of the corresponding vectors is
zero, where $\overline{\mathbf{e}} \equiv \mathbf{u} + \overline\omega
\mathbf{v}$.

With this map, generators of a stabilizer group are mapped to rows of a parity
check matrix $H$ of an \emph{additive} (forming a group with respect to
addition but not necessarily over the full set of $\mathbb{F}_4$ operations)
code over $\mathbb{F}_4$, with the condition that the trace inner
product of any two rows vanishes\cite{Calderbank-1997}.  The vectors generated
by rows of $H$ correspond to stabilizer generators which act trivially on the
code; these vectors form the \emph{degeneracy group} and are omitted from the
distance calculation.  For a more narrow set of Calderbank-Shor-Steane (CSS)
codes the parity check matrix is a direct sum $H=G_x\oplus \omega G_z$, and
the commutativity condition simplifies to $G_{x}G_{z}^{T}=0$.

A LDPC code, quantum or classical, is a code with a
sparce parity check matrix.  For a regular $(j,l)$ LDPC code, every column and
every row of $H$ have weights $j$ and $l$ respectively, while for a
$(j,l)$-limited LDPC code these weigths are limited from above by $j$ and $l$.

The QHPCs \cite{Tillich2009} (Fig.~\ref{fig:Visualization}) are constructed from two binary matrices,
$\mat{H}_{1}$ (dimensions $r_{1}\times n_{1}$) and $\mat{H}_{2}$ (dimensions
$r_{2}\times n_{2}$), as a CSS code with the stabilizer 
\cite{Kovalev-arxiv2012}
\begin{equation}
\begin{array}{c}
{\displaystyle G_{x}=(E_{2}\otimes\mathcal{H}_{1},\mathcal{H}_{2}\otimes E_{1}),}\\
{\displaystyle G_{z}=(\mathcal{H}_{2}^{T}\otimes\widetilde{E}_{1},\widetilde{E}_{2}\otimes\mathcal{H}_{1}^{T}).}
\end{array}\label{eq:Till}
\end{equation}
Here each matrix is composed of two blocks constructed as Kronecker products
(denoted with {}``$\otimes$''), and $E_{i}$ and $\widetilde{E}_{i}$, $i=1,2$,
are unit matrices of dimensions given by $r_{i}$ and $n_{i}$, respectively.
In the original construction\cite{Tillich2009}, given the binary parity check
matrix $\mathcal{H}_1=\mathcal{H}_2^{T}$ of an $(h,v)$-limited classical LDPC
code $[n_\mathrm{c},k_\mathrm{c},d_\mathrm{c}]$, the QHPC~(\ref{eq:Till}) is a
CSS code with the parameters
$[[n=n_\mathrm{c}^2+(n_\mathrm{c}-k_\mathrm{c})^2,k=k_\mathrm{c}^2,
d=d_\mathrm{c}]]$, and column and row weights limited by $j\le \max(h,v)$,
$\ell\le h+v$.  An original classical LDPC code produces a quantum LDPC code,
and the corresponding distance scales as $d\propto n^{1/2}$.

\textsc{Our key observation} is that for LDPC codes, quantum or classical,
large-weight errors are mostly those composed of \emph{disjoint} small-weight
clusters that can be detected or corrected independently.  Indeed, e.g., for a
regular $(j,\ell)$ LDPC code, two random (qu)bits have non-zero values in the
same row with probability $z/n$, where $z\equiv (\ell-1)j$.  Correcting any of
such disjoint errors does not affect the syndrome for the others.

More generally, for a $(j,\ell)$-limited LDPC code, we represent all (qu)bits
as nodes of a graph ${\cal G}_1$ of degree at most $z$: two nodes are
connected by an edge iff there is a row in the parity check matrix which has
non-zero values at both positions.  An error with support in a subset ${\cal
  E}\subseteq V({\cal G}_1)$ of the vertices defines the subgraph ${\cal
  G}_1({\cal E})$ induced by ${\cal E}$.  Generally, we will not make a
distinction between a set of vertices and the corresponding induced subgraph.
In particular, a (connected) cluster in ${\cal E}$ corresponds to a connected
subgraph of ${\cal G}_1({\cal E})$.  Different clusters affect disjoint sets
of rows of the parity check matrix.  This implies the following
\begin{lemma}
  \label{lemma:cluster-detectable}
  For a distance-$d$ LDPC code, any error whose support is a union of
  disconnected clusters on ${\cal G}_1$ of weight $w_i<d$, is
  detectable.
\end{lemma}

In the case of an erasure channel (quantum or classical), we actually know
which (qu)bits are affected.  In known locations, a code of distance $d$ can
correct all errors of weight $d-1$ or smaller.  Therefore, correcting clusters
one-by-one, we can guarantee success if all the clusters have weights $w_i<d$.
It is then obvious that the problem of error correction for an erasure channel
is related to the problem of site percolation on graphs\footnote{This relation
  has been noticed in
  Refs.~\protect\cite{Dennis-Kitaev-Landahl-Preskill-2002,%
    Delfosse-Zemor-2012}}.

For any $(j,\ell)$-limited LDPC code, the vertices of the graph
$\mathcal{G}_1$ have degrees at most $z\equiv(\ell-1)j$.  In what follows, we
will need the cluster size distribution [the probability $n_s^{(p)}(x)$ that
the point $x$ is a member of a cluster of size $s$] below the percolation
threshold.  Although one expects
exponential tail in the cluster size distribution, $n_s(p)\le \exp(-s
g(p))$ for all $s\ge s_0$ and some $g(p)>0$ and $s_0>0$, this can be
violated for sufficiently heterogeneous
graphs\cite{Bandyopadhyay-Steif-Timar-2010}.  Restricting the range of
$p$, we have the following lemma:
\begin{lemma}
  \label{lemma:exponential-clusters}
  For any graph ${\cal G}$ with vertex degrees limited by $z$, the
  site- or bond-percolation cluster size distribution has exponential
  tail for $p<p_0\equiv(z-1)^{-1}$.
\end{lemma}
\begin{proof}
  A size-$s$ cluster containing $x$ on ${\cal G}$, after cutting any loops,
  can be mapped to a size-$s$ cluster on $z$-regular tree ${\cal T}_z$ (Bethe lattice), with
  $x$ mapped to the root.  Such a mapping can only increase the perimeter
  (size of the boundary, i.e., number of sites outside the cluster but
  neighboring with a site inside it).  Any size-$s$ cluster on $\mathcal{T}_z$
  has the perimeter $t_z(s)\equiv s(z-2)+2$; for a cluster on $\mathcal{G}$ we
  have $t\le t_z(s)$.
  Let us now use the standard expression for the cluster size
  distribution, $n_s^{(p)}(x)=\sum_t a_{s,t}(x) p^s (1-p)^t$, where
  $a_{s,t}(x)\ge0$ is the number of $x$-containing site-percolation
  clusters of size $s$ and perimeter $t$.  For $p\le p_0$, 
  \begin{equation}
  (1-p)^t \le (1-p_0)^t {(1-p)^{t_z(s)}\over (1-p_0)^{t_z(s)}},
  \label{eq:p-inequality}
\end{equation}
  which gives the exponential tail
  \begin{equation}
   n_s^{(p)}(x) \le n_s^{(p_0)}(x) 
  {(1-p)^2\over (1-p_0)^2} \alpha_z^s, \;\, \alpha_z\equiv
  {p(1-p)^{z-2}\over p_0(1-p_0)^{z-2} }.\label{eq:exponential-tail}
\end{equation}
since $\alpha_z(p)<1$ for $p<p_0$ and $n_s^{(p_0)}\stackrel{\rm
  def}{\le} 1$.
\end{proof}
Notice that the threshold for exponential tail coincides with the lower
boundary of the (bond) percolation transition for degree-limited
graphs\cite{[{Theorem 1.2 in }]Hofstad-2010}, which is also the lower boundary
for the site percolation transition\cite{Hammersley-1961}; both boundaries are
achieved on $z$-regular tree $\mathcal{T}_z$.

We can now formulate the following
\begin{theorem}
  \label{theorem:erasure-threshold}
  For an infinite family of $(j,l)$-limited LDPC codes, quantum or classical,
  where the distance $d$ scales as a power law at large $n$, $d\ge A
  n^\alpha$, with some $\alpha>0$ and $A>0$, asymptotically certain recovery
  is possible for (qu)bit erasure probabilities $p<p_e$, where $p_e\ge
  p_0=(z-1)^{-1}$ and $z\equiv (\ell-1)j$.  A non-zero threshold $p_e$ also
  exists for such code families with the distance scaling logarithmically at
  large $n$, $d\ge A_0 \ln n$.
\end{theorem}

\begin{proof}
  The conditions match those of Lemmas \ref{lemma:cluster-detectable}
  and \ref{lemma:exponential-clusters}.  For $p<p_0$, we just need to
  ensure that the probability to find a cluster of size $s\ge d$
  anywhere on ${\cal G}_1$ vanishes at large $n$, i.e., $n\sum_{s\ge
    d} n_s^{(p)}/s\to 0$.  The sufficient condition on the distance is
  $d>\ln n/|\ln \alpha_z|$, which is always the case at large $n$ with
  power-law distance, and gives $\alpha_z(p)\le e^{-A_0}$ with
  logarithmically increasing distance.  The latter equation is
  satisfied for small enough $p$ by continuity of $\alpha_z(p)$ since
  $\alpha_z(0)= 0$.
\end{proof}

Given the upper limit on the rate of stabilizer codes in Theorem 3.8 of
Ref.~\onlinecite{Delfosse-Zemor-2012}, we also obtain the limit on the rate of
quantum codes in Theorem \ref{theorem:erasure-threshold}:
\begin{consequence}
  Any family of $(j,\ell)$-limited LDPC quantum codes with power-law scaling
  of the minimum distance with the blocksize $n$, has rate $R$ limited by
  \begin{equation}
    R<1-
    {2}\biggl[{z-1-(z-3)\left(\dfrac{z-2}{z-1}\right)^{\ell-1}}\biggr]^{-1}.
    \label{eq:Rate}
  \end{equation}
\end{consequence}

The situation gets a bit more complicated for the \emph{depolarizing channel}
(memoryless binary symmetric channel in the classical case), where the
positions of the errors are unknown.  A bound on single-(qu)bit error
probability which guarantees almost certain error correction for large codes
is given by

\begin{theorem}
  For an infinite family of $(j,l)$-limited LDPC codes, quantum or classical,
  where the distance $d$ scales as a power law at large $n$, 
  asymptotically certain recovery
  is possible for (qu)bit {\rm depolarizing} probabilities $p<p_d\ge p_1$,
  where $4p_1(1-p_1)= p_0^2 (1-p_0)^{2(z-2)}
  < [e(z-1)]^{-2}$, $p_1<1/2$, and $e$ is the base of the natural logarithm.
  A threshold $p_d>0$ also exists for code families with distance scaling
  logarithmically at large $n$.
  \label{theorem:depolarizing-decoding}
\end{theorem}
\begin{proof}
  The clusters can be irrecoverably misidentified only if there exists a set
  of $s\ge d$ connected vertices on $\mathcal{G}_1$ with $m\ge \lceil
  s/2\rceil$ errors.  We will call such sets violating
  $(s,m)$-sets.  To estimate the probability of encountering such a set, we
  notice that an $(s,m)$ set with an additional error at the perimeter can be
  extended to become an $(s+1,m+1)$ set.  Thus, one only needs to count
  connected sets of size $s\ge d$, with perimeter free of errors.  For $s\ge
  d$, the probability $\tilde n_s^{(p)}(x)$ that one of violating $(s,m)$
  sets includes the point $x$, can be limited as $\tilde
  n_s^{(p)}(x)<f_s^{(p)}(x)$, where
  \begin{equation}
    \label{eq:dbl-edge-limit0}
    f_s^{(p)}(x)\equiv \sum_{t}a_{st}(x)\sum_{m=\lceil s/2\rceil}^u {s\choose m} p^m
    (1-p)^{s-m+t}. 
  \end{equation}
  The sum over $m$ can be limited by $2^s p^{s/2}(1-p)^{s/2+t}$, which gives 
  a bound in terms of the regular cluster size distribution, $f_s^{(p)}\le [4
  (1-p)/p]^{s/2}n_s^{(p)}(x)$.  Using Eq.~(\ref{eq:exponential-tail}), we have
  the condition $4 (1-p)/p \alpha_z^2<1$ to have exponential tail for $\tilde
  n_s^{(p)}(x)$.  This condition is satisfied for $p<p_1$.  The rest of the
  proof follows that of Theorem~\ref{theorem:erasure-threshold}.
\end{proof}

Note that this bound is rather loose as there are many sets which differ just
by error-free regions.  We believe a more accurate general bound should be a
factor of $e^2$ larger, $p_1=[2(z-1)]^{-2}$, as can be obtained by counting
chains on a $z$-regular tree (to reduce multiple counts).

We should also note that the threshold in
Theorem~\ref{theorem:depolarizing-decoding} corresponds to maximum-likelihood
(ML) decoding and is \emph{not} related to a particular decoder, as is
commonly done for LDPC codes.  For a practical approach, one can identify the
putative clusters by first adding all non-zero positions in a parity-check row
corresponding to a non-zero syndrome element, then at each step adding up to
$\ell-1$ positions from each of the non-zero-syndrome rows which include
(qu)bit(s) already in the cluster.  A corrected cluster can be dropped
[Theorem~\ref{theorem:depolarizing-decoding}]; putative clusters which cannot
be corrected individually need to be joined with one or more such cluster(s)
nearby.  Below the percolation threshold the actual clusters will typically
have size of order $s\lesssim \ln n /|\ln \alpha_z|$.  For $p<p_1$
[Theorem~\ref{theorem:depolarizing-decoding}], at large $n$ the number $N$ of
classical operations for exhaustive search of the correct error configuration
in each cluster will scale almost linearly with $n$, namely, $N\sim pn q^s\sim
pn^{1+\ln q/\ln z}$, where $q=2$ for a binary classical and $q=4$ for a
quantum code.

So far our discussion has been limited to idealized performance of the code,
which assumes that the syndrome is measured perfectly.  With qubit measurement
errors, LDPC codes suddenly appear at a disadvantage, as a single-qubit error
accompanied by the errors of some of the stabilizer generators (up to $j$)
which involve this qubit, could remain undetected.  In other words, the
effective distance to such combined errors cannot exceed the \emph{minimum}
column weight plus one.  In order to prevent errors from spreading, either we
have to keep the measurement error small, or we have to combine the 
information from different syndrome measurement cycles.  To help with the
bookkeeping, we constructed an auxiliary classical code combining different
time slices.  The code is based on the parity-check matrix of the original
code, and the repetition codes (in the simplest case) for each measured
syndrome\footnote{A. A. Kovalev and L. P. Pryadko, unpublished (2012).}.  This
generalizes the auxiliary three-dimensional gauge model used in the decoding
of the surface codes\cite{Dennis-Kitaev-Landahl-Preskill-2002}.  The
corresponding graph $\mathcal{G}_1$ has up to $2j$ additional neighbors per
qubit, which corresponds to possible syndrome measurement errors in the two
neighboring layers.  The percolation problem works similarly to that with
perfect measurements.  Overall, the analysis in
Theorem~\ref{theorem:depolarizing-decoding} can be repeated with $z\to
z'\equiv j(\ell+1)$, which in the case of isotropic error probability
$p=p_{\rm meas}$ gives a finite lower bound $p\ge [2e(z'-1)]^{-2}$ for such
\emph{fault-tolerant} measurements.

Of course, combining repeated syndrome measurements is only part of
the story of fault-tolerant implementation of a code.  We implicitly
assumed that the hardware allows for parallel measurement of
non-overlapping stabilizer generators, that this can be done in fixed
time even though the corresponding qubits may not necessarily lie next
to each other, and also that all gates are done fault-tolerantly
\cite{Shor-FT-1996,%
  *DiVincenzo-Shor-1996,*Gottesman:PRA1998,*Steane-FT-1997,%
  *Steane-1999,*Knill-nature-2005}, so that errors do not spread.
With these assumptions, the full syndrome can be measured in
approximately $z$ steps, which would take bounded time independent of
the size of the code.  Further, fault-tolerant operation also implies
some implemetation of logical operators, e.g., as suggested in 
Ref.~\onlinecite{Steane-Ibinson-2005}.

Our results apply to QHPCs and related
codes\cite{Tillich2009,Kovalev-arxiv2012}.  As an example, let us start with a
\emph{random} regular $(h,v)$ LDPC code $\mathcal{H}_1=\mathcal{H}_2^{T}$,
with $h<v$.  The rate of such a code is fixed, $R_\mathrm{c}\equiv
k_\mathrm{c}/n_\mathrm{c}=1-h/v$.  With high probability at large
$n_\mathrm{c}$, the code will have the relative distance in excess of
$\delta_\mathrm{c}$ which satisfies the equation\cite{Litsyn:IEEE2002}
\begin{equation}
H(\delta_{c})+(1-R_{c})p_{v}(R_{c},\delta_{c})\leq0.\label{eq:GV-bound}
\end{equation}
Here $H(\delta_{c})\equiv
-\delta_{c}\ln\delta_{c}-(1-\delta_{c})\ln(1-\delta_{c})$ is the natural
entropy, and
$p_{v}(R_{c},\delta_{c})=\ln\left[(1+y)^{k}/2+(1-y)^{k}/2\right]-\delta_{c}v\ln
y-vH(\delta_{c})$ with $y$ being the only positive root of
$(1+y)^{k-1}+(1-y)^{k-1}=(1-\delta_{c})\left[(1+y)^{k}+(1-y)^{k}\right]$.
Such $[n_{c},k_{c},d_{c}]$ codes produce QHPCs (\ref{eq:Till}) which are
regular $(v,v+h)$ LDPC codes [$z=v(v+h-1)$] with the asymptotic rate
$R\equiv k/n=(v-h)^{2}/(h^{2}+v^{2})$ and the distance scaling as
$d/\sqrt{n}=\delta_{c}v/\sqrt{h^{2}+v^{2}}$ where
$\delta_{c}\equiv\delta_{c}(v,h)$ satisfies Eq.~(\ref{eq:GV-bound}).  For
sufficiently large $n$, all errors can be corrected with certainty for
single-qubit error probabilities below $p_{1}$, see
Theorem~\ref{theorem:depolarizing-decoding}. 
Suppose we need to maintain quantum information for $\mathcal{N}$
QEC cycles with a fault probability less than $P_{f}$. We can crudely
estimate the minimal required blocklength from equation $P_{f}=\mathcal{N}\mathcal{M}f_d^{(p)}(x)d/n$
where $\mathcal{M}$ is the number of syndrome measurements per QEC
cycle ($\mathcal{M}=2$ for toric code and $\mathcal{M}=2v(v+h-1)$
for QHPC). Taking $P_{f}/\mathcal{N}=10^{-9}$, we obtain that QHPCs $[[n,n/25,0.09\sqrt{n}]]$
($v=4$, $h=3$) should have at least $30000$ physical ($1200$ logical)
qubits which is less than in $k/2$ copies of
toric code. More careful estimates for the failure probability
should lead to fewer required physical qubits for QHPCs.

\textsc{In conclusion}, we established a sufficient condition for an LDPC code
family with asymptotically vanishing relative distance to have a finite
probability threshold to correct all errors with certainty, including the case
of syndrome measurement errors.  The result is simple: any LDPC code family
with power-law scaling of the distance with the block size has a finite
threshold, see Theorem \ref{theorem:depolarizing-decoding}.  Existence of such
a threshold is one of \emph{necessary} conditions for achieving
fault-tolerance in a quantum computer, and it should facilitate search for new
efficient quantum codes with less stringent requirements on the number of physical qubits.  In particular, we established the existence of such
a threshold for the finite-rate quantum LDPC codes constructed in the
framework of QHPCs\cite{Tillich2009,Kovalev-arxiv2012}. According to our estimates, QHPCs require lower error threshold compared to toric codes while using 
fewer physical qubits at large code blocklength.

In application to subsystem codes \cite{Bacon:PRA2006,Poulin:PRL2005}, we can
guarantee the existence of a finite threshold only in the case where the
actual stabilizer group (as opposed to the gauge group) has generators of
limited weight.  This is the case, e.g., for color codes \cite{Bombin:PRA2007},
but is not generally the case for Bacon-Shor codes and their generalizations,
see, e.g., Ref.~\onlinecite{Bravyi:PRA2011}.

We have also established the existence of a related threshold when the erroneous (qu)bits are known (erasure channel)
 (Theorem~\ref{theorem:erasure-threshold}), which resulted in an upper
bound for the rate of quantum LDPC codes [see Eq.~(\ref{eq:Rate})].  


We are grateful to I. Dumer and M. Grassl for multiple helpful discussions.
This work was supported in part by the U.S. Army Research Office under Grant
No. W911NF-11-1-0027, and by the NSF under Grant No. 1018935.

\bibliography{MyBIB,qc_all,more_qc}
\end{document}